\documentclass[10pt]{article}
\usepackage{amsthm,amsmath}
\usepackage{amssymb}
\usepackage{tikz}
\usepackage{xfrac}
\usepackage{bm}

\usepackage{mathtools}
\usepackage{enumerate}

\newcommand{\cN} { {\mathcal{N}}}
\newcommand{\cS} { {\mathcal{S}}}

\newcommand{\bN} { {\mathbb{N}}}

\newcommand{\bC} { {\mathbb{C}}}

\newcommand{\bZ} { {\mathbb{Z}}}

\newcommand{\bUpsilon}{{\bm \Upsilon}}

\def\res{\operatorname{res}}

\newtheorem{theorem}{Theorem}
\newtheorem{prop}[theorem]{Proposition}
\newtheorem{thm}[theorem]{Theorem}
\newtheorem{corollary}[theorem]{Corollary}
\newtheorem{lemma}[theorem]{Lemma}
\newtheorem{remark}[theorem]{Remark}
\newtheorem{algorithm}[theorem]{Algorithm}

\newtheorem{defi}[theorem]{Definition}
\newtheorem{example}[theorem]{Example}

\newtheorem{hypo}[theorem]{Hypothesis}

\makeatletter
\newcommand{\myitem}[1]{%
	\item[(#1)]\protected@edef\@currentlabel{#1}%
}
\makeatother

\def\eatspace#1{#1}
\def\step#1#2{\par\kern1pt\hangindent#2em\hangafter=1\noindent\rlap{\small#1}\kern#2em\relax\eatspace}

\let\set\mathbb
\def\<#1>{\langle#1\rangle}

\def\lc{\operatorname{lc}}

\begin{document}

\title{Symbolic Integration in Weierstrass-like Extensions
\thanks{S.\ Chen, W.\ Li and X.\ Li were partially supported by the National Key R\&D Program of China (No. 2023YFA1009401) and
		the NSFC grant (No. 12271511).
		M.\ Kauers was supported by the Austrian FWF grants PAT 9952223, PAT 8258123, and I6130-N.
		X.\ Li was partially supported by the Land Oberösterreich through the LIT-AI Lab.
		This work also was supported by the International Partnership Program of Chinese Academy of Sciences (Grant No. 167GJHZ2023001FN).}
}


\author{
	Shaoshi Chen$^{a, b}$, Manuel Kauers$^c$,  \\
\bigskip
Wenqiao Li$^{a, b}$, Xiuyun Li$^{a, b, c}$, and David Masser$^d$\\
	$^a$KLMM,\, Academy of Mathematics and Systems Science, \\ Chinese Academy of Sciences, \\100190 Beijing, China\\
	$^b$School of Mathematical Sciences, \\University of Chinese Academy of Sciences,\\100049 Beijing, China\\
	$^c$Institute for Algebra, \\Johannes Kepler University,\\ 4040 Linz, Austria\\
    $^d$Department of Mathematics and Computer Science, \\University of Basel, \\ \bigskip Basel, 4003, Switzerland \\
	{\sf schen@amss.ac.cn, manuel.kauers@jku.at}\\
	{\sf liwenqiao@amss.ac.cn,lixiuyun@amss.ac.cn,david.masser@unibas.ch}
}

\maketitle
%\renewcommand{\shortauthors}{Chen et al.}

\begin{abstract}
 This paper studies the integration problem in differential fields
		that may involve quantities reminiscent of the Weierstrass $\wp$
		function, which are defined by a first-order nonlinear differential equation.
		We extend the classical notion of special polynomials to elements of Weierstrass-like extensions
		and present algorithms for reduction in such extensions. As an application of these
		results, we derive some new formulae for integrals of powers of~$\wp$.
\end{abstract}

\maketitle

\section{Introduction}
	
	Modern algorithms for symbolic integration are primarily governed by two complementary approaches. One
	approach uses annihilating linear operators to describe functions, and performs
	computations in operator algebras~\cite{Zeilberger1990,ChyzakSalvy1998,Kauers2023} in order to answer questions
	about a given integral. The other approach uses differential fields to
	describe functions, and performs computations in these fields in
	order to answer questions about a given integral~\cite{Risch1969,Risch1970,BronsteinBook}.
	The latter approach has proved particularly successful for so-called elementary functions,
	but has also seen extensions to other types of functions, see \cite{raab22} and the references given there.
	
	Here we are concerned with integrals involving functions that are defined by
	(possibly nonlinear) first-order differential equations, i.e., functions~$y$
	which satisfy $Q(y',y)=0$ for some bivariate polynomial~$Q$. A prototypical
	example for such a function is the Weierstrass $\wp$ function, which satisfies
	\begin{equation}\label{EQ:dewp}
		\wp'(z)^2=4\wp(z)^3-g_2\wp(z)-g_3
	\end{equation}
	for certain constants $g_2,g_3\in\set C$ with $g_2^3-27g_3^2\neq0$.
	
	In differential algebra aspect, quantities satisfying this differential equation
	were already considered by Kolchin~\cite{Kolchin1953}, who called them
	\emph{Weierstrassian.} Slightly generalizing by also allowing other
	polynomials~$Q$, we call the extensions generated by these functions \emph{Weierstrass-like.} There has been
	some recent work related to symbolic integration with
	Weierstrassian elements. Variants of Liouville's theorem covering this case
	have been proposed by Kumbhakar and Srinivasan~\cite{kubhakar23} and by Pila and
	Tsimerman~\cite{pila22}.
	
	From a more computational perspective, parallel integration has been applied to
	integration problems involving functions defined by nonlinear differential
	equations. Bronstein~\cite{Bronstein2007} gives an example involving the Lambert~$W$
	function, and B\"ottner~\cite{Boettner2010}  gives some examples involving the
	Weierstrass~$\wp$ function. Apart from this, not much is known about
	integration theory and algorithms when the integrand involves $\wp$ or similar
	functions. In fact, not too many identities about integrals involving $\wp$ are
	available in the literature.
	
	Our contribution in this paper consists of three aspects. On the theoretical
	side, we propose an extension of the classical notion of \emph{special
		polynomials} to elements of Weierstrass-like extensions. On the algorithmic
	side, we propose an algorithm for reduction process, thereby
	continuing an ongoing trend in the development of integration
	algorithms~\cite{bostan10b,BCCLX2013a,Chen2018JSC,ChenKauersKoutschan2016,DHL2018,BCLS2018,chen21b,vanderHoeven21,CDK2023,DGLL2025}.
	Finally, as an application of these
	results, we obtain some new identities about integrals of powers of~$\wp$.
	
	\section{An appetizer}\label{SECT:start}
	We start our study from the integration problem of powers of the Weierstrass~$\wp$ function, i.e.,
	the problem of evaluating the integral $I_n(z) :=  \int \wp(z)^n \, dz$ with $n\in \bZ$.
	We will recall some classical formulae from the book~\cite{WW1927}.
	By differentiating the both sides of the differential equation~\eqref{EQ:dewp}, we obtain that
	$\wp(z)^2 = \frac{1}{6}\wp''(z) + \frac{1}{12}g_2$,
	which implies that
	\[\int \wp(z)^2 \, dz = \frac{1}{6}\wp'(z) + \frac{1}{12}g_2z + C, \quad \text{$C$ is a constant}. \]
	In the following, we will always omit $C$ if no confusion arises.
	So the integral $I_2(z)$ is in the field generated by $\wp(z)$ and $\wp'(z)$ over $\bC(z)$.
	After developing the theory of integration in elementary terms in this paper, we will show that the integral
	of $\wp(z)$  is not elementary over the field $\bC(z)(\wp(z), \wp'(z))$. For this reason, we should introduce some new functions, such as
	the Weierstrass zeta-function $\zeta(z)$ satisfying the equation $\zeta'(z) = -\wp(z)$ and the function $\sigma(z)$ with $\zeta(z) = \sigma'(z)/\sigma(z)$.
	In terms of these two new functions, we can evaluate the integral
	\[\int \frac{1}{\wp(z)}\, dz = \frac{1}{\sqrt{-g_3}} \left(\log\left(\frac{\sigma(z-v)}{\sigma(z+v)}\right) + 2\zeta(v) z\right),\]
	where $v\in \bC$ be such that $\wp(v)=0$ and $\wp'(v)= \sqrt{-g_3}$. To evaluate the general integral
	$I_n(z)$, we need the inverse function of $\wp(z)$, which is defined by the formula~\cite[p.\ 438]{WW1927}
	\[\wp^{-1}(z) = \int_{z}^{\infty} \frac{1}{\sqrt{4t^3- g_2t - g_3}} \, dt.\]
	By using change of variables, we have
	\[\int \wp(z)^n \, dz = J_n(\wp(z)), \quad \text{with $J_n(t) := \int \frac{t^n\, dt}{\sqrt{4t^3- g_2t - g_3}}$}.\]
	The integral $J_n(t)$ satisfies a linear recurrence equation of the form
	
	\begin{alignat*}1
		J_{n}(t) &= \frac{g_2(2n-3)}{8n-4} J_{n-2}(t) + \frac{g_3(n-2)}{4n-2} J_{n-3}(t)\\
		&+ \frac{t^{n-2}\sqrt{4t^3- g_2t - g_3}}{4n-2},
	\end{alignat*}
	with initial values $J_0(t) = \wp^{-1}(t), J_1(t) = -\zeta(\wp^{-1}(t))$ and
	\[ J_2(t) = \frac{\sqrt{4t^3- g_2t - g_3}}{6} + \frac{g_2}{12} \wp^{-1}(t).\]
	The above recurrence can be computed by creative telescoping using the {\sc Mathematica} package {\bf HolonomicFunctions.m}~\cite{koutschan10c}.
	From the above recurrence for $J_n(t)$ and changing back the variables, we can obtain the integral
	\[\int \wp(z)^3 \, dz = \frac{1}{10} \wp(z)\wp'(z) - \frac{3g_2}{20} \zeta(z) +\frac{g_3}{10} z \]
	and also the integral
	\[\int \wp(z)^4 \, dz = \frac{1}{14} \wp(z)^2\wp'(z)+ \frac{5g_2}{168}\wp'(z)  + \frac{5g_2^2}{336} z -\frac{g_3}{7}\zeta(z) .\]
	We will revisit these integrals $I_n(z)$ with $n\in \bN$ in Section~\ref{SECT:revisit}.
	\section{Algebraic Functions}\label{SECT:Pre}
	
	In this section, we recall some terminologies about fields of algebraic functions
	of one variable from the book~\cite{chevalley1951}. Let $k$ be a field of
	characteristic $0$ and $t$ be transcendental over~$k$. If $m\in
	k[t,X]$ is an absolutely irreducible polynomial, i.e., $m$ is
	irreducible over the algebraic closure $\bar{k}$ of $k$, then $k$ is algebraically closed in $K:=k(t)[X]/\<m>$ by~\cite[Section 3, Thm 1]{Trager1984}. Such $K$ is called a \emph{field of algebraic
	functions of one variable} over~$k$.
	
	A subring $\mathcal{O}\subsetneq K$ is called a {\em valuation ring} if $\mathcal{O}$ contains $k$, and for any
	nonzero $x\in K$, either $x\in \mathcal{O}$ or $x^{-1}\in \mathcal{O}$.  A
	valuation ring is a local ring with a principal maximal ideal. The maximal ideal
	of a valuation ring is called a \emph{place.}
	For a given place~$P$, there is a unique valuation ring $\mathcal{O}_P$ of which
	$P$ is the maximal ideal. The residue field $\mathcal{O}_P/P$ is denoted by
	$\Sigma_P$. It is a finite algebraic extension of~$k$.
	
	Assume $P=u\mathcal{O}_P$. The \emph{order function} of $K$ at $P$ is a
	map $\nu_P\colon K \to \bZ \cup \{+\infty\}$ defined by $\nu_{P}(x)=\max\{n\mid x\in u^n\mathcal{O}_P \}$.
	One can prove that $\nu_{P}$ is well-defined and has the following properties:
	\begin{itemize}
		\item[(i)] $\nu_P(x)=+\infty$ if and only if $x=0$;
		\item[(ii)] $\nu_P(xy)=\nu_P(x)+\nu_P(y)$ for all $x,y\in K$;
		\item[(iii)] $\nu_P(x+y)\geq\min\{\nu_P(x),\nu_P(y)\}$ for all $x,y\in K$ and equality holds if $\nu_{P}(x)\ne \nu_{P}(y)$
	\end{itemize}
    In fact,
	$\mathcal{O}_P := \{x\in K\mid\nu_P(x)\ge0\}$,
	and every elements of $P$ have strictly positive order at $P$.
	
	\begin{example}\label{EX:rat}
		Let $k$ be a field of characteristic $0$ and $k(t)$ be the field of rational
		functions over~$k$. A place of $k(t)$ is generated by either $t^{-1}$ or an
		irreducible polynomial $p\in k[t]$ in their valuation ring
		$\mathcal{O}_{\infty}$ or $\mathcal{O}_p$, respectively. We denote the
		respective order functions by $\nu_{\infty}$ and $\nu_{p}$. 
		Let $a,b\in k[t]$ with $ab\neq 0$, and $f=a/b\in C(t)$.  Then
		$\nu_{\infty}(f)=\deg_t b-\deg_t a$.  If we write $f=p^{n}a_0/b_0$ where $n\in
		\mathbb{Z}$, $\gcd(a_0,b_0)=1$ and $p\nmid a_0b_0$, then $\nu_{p}(f)=n$. We say that
		$t^{-1}\mathcal{O}_{\infty}$ is the {\em infinite place} of $k(t)$ and
		$p\mathcal{O}_p$ is a {\em finite place}.
		In fact, 
		\[
		\mathcal{O}_\infty=\left\{\frac{a}{b}\mid \deg_t a< \deg_t b \right\}\,\, \text{and}\,\,\mathcal{O}_p= \left\{\frac{a}{b} \mid \ \gcd(a,b)=1,\ p\nmid b\right\}
		\]

		Let $\Sigma_\infty$ be the residue field of $\mathcal{O}_\infty$ and $\Sigma_p$
		be that of $\mathcal{O}_p$. It is straightforward to check that
		$\Sigma_\infty=k$, and $\Sigma_{p}$ is isomorphic to $k(\beta)$, where
		$\beta\in \bar{k}$ is a root of~$p$. 
		%Consequently, taking the image of an element in $\mathcal{O}_{p}$ under the quotient is equivalent to evaluating
		%at~$\beta$.
		
		For convenience, we also use $\mathcal{O}_{t^{-1}}$ to refer to
		$\mathcal{O}_{\infty}$.

	\end{example}
	
	According to \cite[page 119]{Bronstein1990}, for a place $P$ of $K$, either there exists a unique irreducible polynomial $p$ such that $p\in P$, in which case $P$ will be called a \emph{finite place}, or $t^{-1}\in P$, in which case $P$ will be called an \emph{infinite place}. 
	
	Let $q$ be $t^{-1}$ or an irreducible polynomial of $k[t]$. If $q\in P$, then we say that $P$ \emph{lies above}~$q$, or equivalently, $q$~\emph{lies below}~$P$. In fact, $q$ lies below at least one place of $K$ but not infinitely many places. Now assume that $q$ lies below $P$.
	Let $\mathcal{D}$ be the valuation ring of $P$. Then $\mathcal{O}:=\mathcal{D}\cap
	k(t)$ is the valuation ring in $k(t)$ of the place $Q:=P\cap \mathcal{O}$, which is generated by $q$. Since $Q$ is the contraction of $P$ in $\mathcal{O}$, one can identify the residue field $\Sigma_{Q}$ as a subfield of $\Sigma_P$.
	
	%Assume that $P$ is a place of $K$ with
	%valuation ring $\mathcal{D}$. It is immediate that $\mathcal{O}:=\mathcal{D}\cap
	%k(t)$ is a valuation ring in $k(t)$ with place $Q:=P\cap \mathcal{O}$. Assume that $Q$~is generated by~$q$, which is either an irreducible
	%polynomial in $k[t]$ or~$t^{-1}$. Then we say the place $P$ \emph{lies above}~$q$,
	%or equivalently, $q$~\emph{lies below}~$P$. Note that $Q$ is the contraction of $P$
	%	in~$\mathcal{O}$.  One can identify the residue field $\Sigma_{Q}$ as a subfield
	%	of $\Sigma_P$. For an irreducible polynomial in $k[t]$ or $t^{-1}$, it lies below a finite number of places of $K$. If $P$ lies above $t^{-1}$, then we say $P$ is an \emph{infinite
		%	place} of~$K$, otherwise we say $P$ is a \emph{finite place.}
	
	Let $\nu_Q$ be the order function of $k(t)$ at $Q$ and $\nu_{P}$ be that of $K$
	at~$P$. Since $\nu_{P}\big(k(t)\setminus\{0\}\big)$ forms an additive subgroup
	of $\mathbb{Z}$ and contains nonzero integers, it is generated by a positive integer~$r_P$, which is called the {\em ramification index} of~$P$.
	It is clear that
	\begin{equation}\label{EQ:val}
		\forall\ x\in k(t)\setminus\{0\},\, \nu_{P}(x)=r_P\nu_Q(x).
	\end{equation}
	In particular, $\nu_{P}(q)=r_P$.
	
	An element $x\in K$ is said to be \emph{integral} at $q$ if the order of $x$ is
	nonnegative at each place of $K$ lying above~$q$. The set of all elements integral at $q$ is a free $\mathcal{O}_q$-module of rank $[K:k(t)]$. A basis of this module is called a \emph{local integral basis}
	at~$q$. If $x$ is integral at each irreducible polynomial in $k[t]$, we say that
	$x$ is integral over $k[t]$. This is equivalent to saying that the monic minimal polynomial of $x$ belongs to $k[t,X]$. The set of all integral elements is
	a free $k[t]$-module of rank $[K:k(t)]$, and a basis of this module is called an
	\emph{integral basis} of $K$. Several algorithms are known for computing an integral
	basis for $K$~\cite{Trager1984,Rybowicz91,vanHoeij94}.
	
	The following conventions will be used throughout this paper. Let
	$v\in k[t]$ be squarefree and let $S$ be the set of all irreducible
	factors of~$v$. By
	``integral at~$v$'' we mean ``integral at each $p\in S$'', and similarly, by
	``local integral basis at~$v$'' we mean ``local integral basis at each $p\in S$ ''.
	Denote $\mathcal{O}_v=\bigcap_{p\in S}\mathcal{O}_p$. All the elements in $K$, that are integral at $v$, form a
	free $\mathcal{O}_v$-module of rank $[K:k(t)]$.

	\section{Special Polynomials and Places}
	In the rest of this paper, we
	let $(k,\,')$ be a differential field of characteristic~$0$, $C_k$ be its subfield of constants and $k'$ its set of the derivatives in $k$. Let $E$ be a differential extension of~$k$. Kolchin in \cite[page 803]{Kolchin1953} defines that an element $t\in E$ is called \emph{Weierstrassian} over $k$ if $t$ is not a constant and $(t')^2=\alpha^2(4t^3-g_2t-g_3)$, where $\alpha\in k$, $g_2,g_3\in C_k$ and $27g_3^2-g_2^3\ne 0$. In this section, we consider a more general situation.
	
	\begin{defi}\label{DEF:Weierstrass-like}
		Suppose that $t\in E$ is
		transcendental over~$k$ and $t'\in \overline{k(t)}$ is
		integral over $k[t]$. Let $m(t,X)\in k[t,X]$
		be the monic absolutely irreducible minimal polynomial of~$t'$. Then $K:=k(t,t')$ is a field of algebraic functions of one variable, and a differential extension of $(k,\,')$.  We call $K$ a \emph{ Weierstrass-like extension} of~$k$.
	\end{defi}
	
	If $t$ is transcendental and Weierstrassian over $k$, then $k(t,t')$ is a Weierstrass-like extension. Definition~\ref{DEF:Weierstrass-like} also generalizes Bronstein's notion of monomial extensions~\cite[Section 3.4]{BronsteinBook}. A monomial extension is a
	Weierstrass-like extension with $m(t,X)= X-s(t)$ for some $s \in k[t]$. In a monomial extension, a
	polynomial $p\in k[t]$ is called special if $p\mid p'$. Being special is closely
	related to $K$ having more constants than~$k$, and with differentiation
	affecting orders in a different way than the usual derivation
	$\frac{d}{dx}$. That is why special polynomials need special attention.
	
	In our case, we have nontrivial algebraic extensions, so $p'$ need not be a
	polynomial even if $p$ is. We therefore need to refine the concept of being
	special for Weierstrass-like extensions.
	
	\begin{defi}\label{DEF:special}
		Let $\beta \in \bar{k}$. We call $\beta$ a \emph{special point} of $m$
		if $m(\beta, \beta') = 0$. For an irreducible polynomial $p\in k[t]$, we say
		that $p$ is \emph{special} (w.r.t.~$m$) if it admits a root that
		is a special point. Otherwise, $p$ is called \emph{normal}.
	\end{defi}	
	
	By \cite[Thm 3.2.4]{BronsteinBook}, a $k$-automorphism of any
	algebraic extension of $k$ commutes with derivation. Hence conjugation preserves  the specialness. Therefore, if one of the
	roots of a polynomial $p$ is a special point, then all roots of $p$ are special
	points. The definition of special points corresponds to~\cite[Thm
	3.4.3]{BronsteinBook}.
	
	For $q\in k[t]$, we let $\kappa(q)$ be the polynomial obtained by differentiating the coefficients of~$q$, and $\partial_t(q) = \frac{d}{dt}(q)$ be the formal derivative of $q$ with
	respect to~$t$. The operations $\kappa$ and $\partial_t$ are both derivations on $k[t]$, and $q'=\kappa(q)+\partial_t(q)t'$. One can see that $q'$ is integral over $k[t]$ since $t'$ is.
	
	\begin{thm}\label{THM:vspecial}
		Let $p \in k[t]$ be irreducible. Then $p$ is special if and only
		if there exists a place $P$ of $K$ lying above $p$ such that $\nu_{P}(p') > 0$.
	\end{thm}
	\begin{proof}
		Let $\beta \in \bar{k}$ be a
		root of~$p$. Assume that $p = (t-\beta)\tilde{p}$, where $\tilde{p}\in
		k(\beta)[t]$ and $\tilde{p}(\beta) \neq 0$.  Then:
		\begin{equation}\label{EQ:diffroot}
			\kappa(p)= - \beta'\tilde{p} + (t-\beta)\kappa(\tilde{p}),
			\quad
			\partial_t(p)= \tilde{p} + (t-\beta)\partial_t(\tilde{p}).
		\end{equation}
		
		Recall that $t'$ is integral over $k[t]$. For any place $P$ lying above $p$, passing to the residue field $\Sigma_P$ yields
		\[
		\overline{p'}=\overline{\kappa(p)}+\overline{\partial_t(p)}\,\overline{t'}.
		\]
		Since $\Sigma_p$ can be considered as a subfield of $\Sigma_P$ through
		$\Sigma_p \cong k(\beta)$, it follows that
		\[
		\overline{p'}= \tilde{p}(\beta)\bigl(\overline{t'} - \beta'\bigr).
		\]	
		Hence $\nu_{P}(p')>0$ if and only if $\overline{t'}=\beta'\in \Sigma_P$.
		
		At first, assume that $\nu_P(p') > 0$ for some
		place $P$ lying above~$p$. Then $\overline{t'} = \beta'$. Since
		$m(t,t')= 0$, taking images into the residue field yields
		$m(\overline{t},\overline{t'})= \bar{0}$, hence $m(\beta, \beta')= 0$,
		i.e., $p$ is special.
		
		Conversely, we assume that $p$ is special. Then $m(\beta,\beta') = 0$, so $\beta'$ is a root of $m(\beta, X)$. Assume $m(\beta, X) = (
		X-\beta')^{s}\,\tilde{m}( X)$, where $\tilde{m}(
		X)\in k(\beta)[X]$ and $\tilde{m}(\beta') \neq 0$.
		By~\cite[Thm 3.3.7]{stichtenoth2009}, there is a place $P$ of $K$ lying
		above $p$ such that $\overline{t'} = \beta' \in \Sigma_P$. Thus $\nu_{P}(p') > 0$.
	\end{proof}
	This theorem indicates that Definition~\ref{DEF:special} generalizes the notion of
	special polynomials in~\cite[Section 3.4]{BronsteinBook}.	
	
	\begin{defi}
	A polynomial $q\in k[t]$ is called \emph{normal} if $q$ is squarefree and all
		irreducible factors of $q$ are normal.
	\end{defi}
	
	\begin{prop}\label{PROP:normal}
		Let $q\in k[t]$ be a normal polynomial and $p$ be an irreducible factor of~$q$.
		Then $\nu_{P}(q')=0$ for every place $P$ lying above~$p$.
	\end{prop}
	\begin{proof}
		Write $q=p\tilde{q}$ with $\gcd(p,\tilde{q})=1$. Then
		$q'=p'\tilde{q}+p\tilde{q}'$. For any place $P$ of $K$ lying above~$p$,
		$\nu_P(p')=0$ by Theorem~\ref{THM:vspecial}. So
		$\nu_P(p'\tilde{q})=0$. Note that $\nu_{P}(p\tilde{q}')>0$ since $\tilde{q}'$ is integral over $k[t]$. Hence $\nu_{P}(q')=0$.
	\end{proof}

	We now introduce the notion of special places in~$K$.
	
	\begin{defi}
		Let $P$ be a finite place of $K$ lying above $p\in k[t]$. We call $P$ a \emph{special place} if
		$\nu_{P}(p') > 0$, and \emph{special of the zeroth kind} if $0<\nu_{P}(p')<\nu_{P}(p)$. We call that $P$ is \emph{normal} if $\nu_P(p') = 0$.
	\end{defi}
	
	Extending Bronstein’s~\cite{BronsteinBook} theory of special polynomials, we encounter places of the zeroth kind, which cannot appear in monomial extensions. The name ``the zeroth kind '' is in contrast to Bronstein’s first-kind case $\nu_P(p')=\nu_P(p)$, because the zeroth-kind case corresponds to $\nu_{P}(p')<\nu_{P}(p)$. Places of the zeroth kind allow effective control of orders after differentiation and are the only ones involved in the special reduction to be developed in Section~\ref{SECT:CT}.
	%which we call \emph{special of the zeroth kind}, in contrast to Bronstein’s first-kind case $\nu_P(f')=\nu_P(f)$.
	If $P$ is a special place lying above~$p$, then $p$ is also special by Theorem~\ref{THM:vspecial}.

	\begin{prop}\label{PROP:order}
		Let $P$ be a finite place of $K$ lying above $p\in k[t]$, and let $r$ be the ramification index of~$P$.
		Then, for any $f \in K\setminus\{0\}$:
		\begin{itemize}
			\item[(i)] If $\nu_P(f)= 0$, then $\nu_{P}(f') \ge  \min\{0,\nu_P(p')-r+1\}$.
			\item[(ii)] If $\nu_P(f)\ne 0$, then $\nu_{P}(f') \ge \min\{\nu_{P}(f), \nu_{P}(f) +\nu_P(p') - r \}$, and equality holds if $P$ is normal or special of the zeroth kind.
			
		\end{itemize}
	\end{prop}
	\begin{proof}
		Note that $\nu_{P}(p)=r$ by \eqref{EQ:val}. Let $v = \nu_{P}(f)$. Assume that $f$ admits a
		local expansion of the form
		$f = b_{0} p^{\frac{v}{r}} + b_{1} p^{\frac{v+1}{r}} + \cdots$,
		where $b_i\in \bar{k}$ and $b_{0} \neq 0$. Differentiating $f$ gives
		\[
		f'
		= b_{0}' p^{\frac{v}{r}} + b_{1}' p^{\frac{v+1}{r}} + \cdots
		+ p'\left(
		\frac{v}{r} b_{0} p^{\frac{v-r}{r}}
		+ \frac{v+1}{r} b_{1} p^{\frac{v+1-r}{r}}
		+ \cdots
		\right) .
		\]
		
		(i) If $v=0$, then the order of $f'$ is no less than the respective orders of $p^0$ and $p'p^{\frac{1-r}{r}}$.
		%	 Possible nonzero terms of the
		%	lowest order are between $p^{0}$ and $p'p^{\frac{1-r}{r}}$.
		
		(ii) If $v\ne 0$, then the order of $f'$ is no less than the respective orders of $p^{\frac{v}{r}}$ and $p'p^{\frac{v-r}{r}}$.
		%	 possible terms of the
		%	lowest order are between $p^{\frac{v}{r}}$ and $p'p^{\frac{v-r}{r}}$.
		Hence $\nu_{P}(f')\ge \min\{v,v+\nu_P(p')-r \}$. If $P$ is normal or special of the zeroth kind, then $\nu_P(p')<r$. Therefore,  $p'p^{\frac{v-r}{r}}$ dominates the order of $f'$, because its order is strictly less than $p^{\frac{v}{r}}$. Thus $\nu_{P}(f')=v+\nu_{P}(p')-r$.
	\end{proof}
	
	\begin{corollary}\label{COR:constant}
		If all special places of $K$ are of the zeroth kind, then the field of constants of $K$ coincides with $C_k$.
	\end{corollary}
	\begin{proof}
		Let $c\in K$ be a constant, i.e., $c'=0.$ Assume $c\notin k$, then there is a place $P$ of $K$ such that $\nu_{P}(c)<0$ by~\cite[page 9, Corollary 3]{chevalley1951}. Note that $P$ is either normal or special of the zeroth kind. It follows from Proposition~\ref{PROP:order} that $\nu_{P}(0)=\nu_{P}(c')$ is a finite number, a contradiction.
	\end{proof}
	
	We also need an analogue of Proposition~\ref{PROP:order} at infinite places as a preparation for studying elementary integrablity in Section~\ref{SECT:CT}.
	
	\begin{prop}\label{PROP:orderinf}
		Let $P$ be an infinite place of $K$ with ramification index~$r$. Then, for any $f\in K\setminus\{0 \}$:
		\begin{itemize}
			\item[(i)] If $\nu_{P}(f)=0$, then $\nu_{P}(f')\ge \min\{0,\nu_{P}(t')+r+1\}$.
			\item[(ii)] If $\nu_{P}(f)\ne 0$, then $\nu_{P}(f')\ge \min\{\nu_{P}(f),\nu_{P}(f)+\nu_{P}(t') +r\}$.
		\end{itemize}
	\end{prop}
	\begin{proof}
		We have $\nu_{P}(t)=-r$ by \eqref{EQ:val}. Let $v=\nu_{P}(f)$. Assume $f$ admits a local expansion at infinity:
		$f=b_0t^{\frac{-v}{r}}+b_1t^{\frac{-v-1}{r}}+\cdots$,
		where $b_i\in \bar{k}$ and $b_0\ne 0$. Differentiating $f$ yields
		\[
		f'
		= b_{0}' t^{\frac{-v}{r}} + b_{1}' t^{\frac{-v-1}{r}} + {\cdots}
		+t'\left(
		\text{\footnotesize $\frac{-v}{r}$}
		b_{0} t^{\frac{-v-r}{r}}
		{+} \text{\footnotesize $\frac{-v-1}{r} b_{1}$}
		t^{\frac{-v-1-r}{r}}
		+ {\cdots}
		\right),
		\]
		from which one can derive assertions (i) and (ii) by a similar argument as in the proof of Proposition~\ref{PROP:order}.
	\end{proof}
	
	\section{Hermite Reduction}\label{SECT:HT}
	
	A core algorithmic technique in symbolic integration is Hermite reduction, which
	decomposes integrands into an integrable part and a
	remainder with controlled poles. It is extended from
	rational functions~\cite{Ostrogradsky1845, Hermite1872} to transcendental elementary functions by
	Risch~\cite{Risch1969, Bronstein1990, BronsteinBook}, and
	to algebraic functions by Trager~\cite{Trager1984, Bronstein1990}
	via integral bases. More recently, Hermite reduction has been
	generalized to D-finite functions~\cite{BCLS2018,Chen2018JSC,vanderHoeven21,CDK2023}.
	
	Throughout this section, let $K=k(t,t')$ be a Weierstrass-like
	extension over~$k$, and let $m\in k[t,X]$ be the monic minimal polynomial of $t'$ with degree $n$ in $X$. Let $\{\omega_1,\dots,\omega_n\}$ be an integral basis of~$K$.
	Then an element $f\in K$ can be written as
	$f=\sum_{i=1}^{n}\frac{f_i}{D}\omega_i$ for some $D,f_i\in k[t]$ with
	$\gcd(D,f_1,\dots,f_n)=1$. Such $D$ is unique up to a nonzero multiplicative element of $k$. We call $D$ the \emph{denominator} of $f$ w.r.t. $\{\omega_1,\dots,\omega_n\}$.
	Write $D=D_ND_S$, where all the irreducible factors of $D_N$ are normal and those of $D_S$ are special.
	Note that $D_N$ and $D_S$ are coprime. See Appendix \ref{APPEN:splitting} for an algorithm to compute $D_N$ and $D_S$ without irreducible factorization.
	
	By the extended Euclidean algorithm, one can uniquely decompose $f$ as $\cN (f)+\cS (f)$,
	where
	\[
	\cN (f)=\sum_{i=1}^{n}\frac{a_i}{D_N}\omega_i,\,\, \text{and}\,\, \cS (f)=\sum_{i=1}^{n}\frac{b_i}{D_S}\omega_i
	\]
	for some $a_i,b_i\in k[t]$ with $\deg(a_i)<\deg(D_N)$. We call $\cN (f)$ and $\cS (f)$ the \emph{normal and special parts} of~$f$, respectively. We say that
	$\cN (f)+\cS (f)$ is the \emph{canonical representation} of $f$ w.r.t. $\{\omega_1,
	\dots, \omega_n\}$. Both $\cN$ and $\cS$ can be regarded as $k$-linear operators on $K$.
	
	The idea of Hermite reduction is to decrease the multiplicity of factors of $D_N$ modulo derivatives.
	We begin with a technical lemma analogous to a result in~\cite[Section
	4.2]{Trager1984}. It will be used later to guarantee the
	correctness of Hermite reduction.
	
	\begin{lemma}\label{LEM: localbasis}
		Let $v\in k[t]$ be a normal polynomial and $\mu > 1$ be an integer.
		Set $\psi_i :=  v^{\mu}\bigl(v^{1-\mu}\omega_i\bigr)'$ with $i=1,\dots,n$.
		Then $\{\psi_1,\ldots,\psi_n\}$ is a local integral basis at $v$.
	\end{lemma}
	\begin{proof}
		Let $p$ be an arbitrary irreducible factor of~$v$, and let $Q$ be a place of $K$
		lying above $p$ with ramification index $r_Q$. Then $Q$ is normal.  Since each $\omega_i$ is
		integral over $k[t]$, $\nu_Q(\omega_i')>-r_Q$ by Proposition~\ref{PROP:order}. Note that $v$ is normal. Then $\nu_Q(v')=0$
		by Proposition~\ref{PROP:normal}. Therefore, $\nu_Q(\psi_i)\ge0$ because $\psi_i=v\omega_i'-(\mu-1)v'\omega_i$. Consequently,
		$\psi_i$ is integral at~$p$. It follows that $\psi_i$ is integral at~$v$.
		
		There are two ways how $\{\psi_1,\ldots,\psi_n\}$ can fail to be
		a local integral basis at $v$: (i) $\psi_1,\ldots,\psi_n$ are
		$\mathcal{O}_v$-linearly dependent; (ii) $\psi_1,\ldots,\psi_n$ are
		$\mathcal{O}_v$-linearly independent, but there exists an element
		integral at $v$ which is not a $\mathcal{O}_v$-linear combination of
		$\psi_1,\ldots,\psi_n$. In both cases, there is an $F\in K$, integral at~$v$,
		such that $F=\frac{1}{v}\sum_{i=1}^{n}c_i\psi_i$, where $c_1,\dots,c_n\in
		k[t]$ are not all zero and $v \nmid c_j$ for some $j$. We
		will derive a contradiction from the existence of such element~$F$.
		
		Let $G = \sum_{i=1}^{n} c_i' \omega_i$. Note that $c_i'$ is integral over $k[t]$ since $c_i\in k[t]$, so is $G$. We have:
		\begin{equation}\label{EQ:integral}
			F + G
			= v^{\mu-1}\sum_{i=1}^{n}\!
			\left( c_i (v^{1-\mu}\omega_i)' + c_i' v^{1-\mu}\omega_i \right)
			= v^{\mu-1}\!\sum_{i=1}^{n} (c_i v^{1-\mu}\omega_i)' .
		\end{equation}
		Let $H = \sum_{i=1}^{n} c_i v^{1-\mu}\omega_i=v^{2-\mu}\sum_{i=1}^{n}\frac{c_i}{v} \omega_i$. Then $H\ne
		0$. Since $v\nmid c_j$, some irreducible factor $p$ of $v$ appears in the denominator of $\sum_{i=1}^{n}\frac{c_i}{v} \omega_i$. Then there exists a place $P$ of $K$
		lying above $p$ such that $\nu_P(H) <
		\nu_P\bigl(v^{2-\mu}\bigr)$, where $\nu_P$ denotes the order function at~$P$. Let
		$r_P$ be the ramification index of~$P$. Then $\nu_{P}(H)<(2-\mu)r_P\le 0$
		by $\mu>1$. Then $\nu_P(H')<(1-\mu)r_P=\nu_P\bigl(v^{1-\mu}\bigr)$ by Proposition~\ref{PROP:order}.
		
		However, $H' = v^{1-\mu}(F+G)$ by~\eqref{EQ:integral}, which yields a contradiction since $F+G$ is integral at $v$.
	\end{proof}
	
	We now describe the Hermite reduction in $K$.
	For convenience, assume that $f=\cN (f)=\sum_{i=1}^{n}\frac{f_i}{D}\omega_i\in K$.
	Then all irreducible factors of $D$ are normal.
	
	Let $D=u v^{\mu}$, where $\mu>1$ is an integer, $v$ is squarefree, $\gcd(u,v)=1$ and all irreducible factors of $u$ have multiplicities less than $\mu$. Set $\tilde f:=fD$, which is integral over $k[t]$, and
	define $\psi_i:=\bigl(v^{1-\mu}\omega_i\bigr)'D$.  By Lemma~\ref{LEM:
		localbasis}, $\{\psi_1,\dots,\psi_n\}$ is a local integral basis at~$v$.
	Let us decrease the multiplicity of $v$. We
	compute $c_1,\dots,c_n\in k(t)$ such that $\tilde{f}=\sum_{i=1}^{n}c_i\psi_i$.
	Then $c_i\in \mathcal{O}_v$. For each~$i$, we can find $r_i\in k[t]$ such that $\deg_t r_i<\deg_t v$ and $r_i\equiv c_i \mod v$. Set $\tilde{g}=\sum_{i=1}^{n}r_i\psi_i$. Then $\tilde{f}-\tilde{g}=vR$ for some $R\in K$ which is integral at~$v$. Hence
	\[
	f=\frac{\tilde{g}+vR}{D}=\sum_{i=1}^{n}r_i(v^{1-\mu}\omega_i)'+\frac{R}{uv^{\mu-1}}.
	\]
	Using integration by parts,
	\begin{equation}\label{EQ:reduction}
		f=\sum_{i=1}^{n}\Bigl(\frac{r_i}{v^{\mu-1}}\omega_i\Bigr)'+\frac{R}{uv^{\mu-1}}-\sum_{i=1}^{n}\frac{r_i'}{v^{\mu-1}}\omega_i.
	\end{equation}
	Let $g=\sum_{i=1}^{n}\frac{r_i}{v^{\mu-1}}\omega_i$. Since $r_i'$ is integral
	over $k[t]$, the denominator of $f-g'$
	w.r.t. $\{\omega_1,\dots,\omega_n\}$ has multiplicity less than $\mu$ at~$v$ by \eqref{EQ:reduction}.
	
	However, the derivative $g'$ may introduce new factors of the denominator since $\omega_i'$ does
	not need to be integral over $k[t]$. Denote $\vec{\omega}=(\omega_1,\dots,\omega_n)^{\tau}$, where $\tau$ denotes the transpose of a vector.
	Let $e\in k[t]$ and $M=(m_{i,j})_{i,j=1}^n\in k[t]^{n\times n}$ be such that
	\[
	e\bigl(\vec{\omega}\bigr)'=M\vec{\omega}
	\]
	with $\gcd(e, m_{1,1}, \dots, m_{n,n}) = 1$. Then $e$ is unique up to a nonzero multiplicative element of $k$. We call $e$ the \emph{differential denominator} of $\{\omega_1,\dots,\omega_n\}$. Denote
	$\vec{r}=(r_1,\dots,r_n)^{\tau}$.
	Then $g=\frac{\vec{r}^{\tau}\vec{\omega}}{v^{\mu-1}}$. A direct calculation shows
	that
	\begin{equation}\label{EQ:HermiteG}
		g'=\underbrace{\frac{\bigl(\vec{r}^{\tau}\bigr)' \vec{\omega}}{v^{\mu-1}}+\frac{(1-\mu)v'\vec{r}^{\tau}\vec{\omega}}{v^{\mu}}}_{A}+\frac{\vec{r}^{\tau} M \vec{\omega}}{ev^{\mu-1}},
	\end{equation}
	in which $e$ may introduce new factors of the denominator.
	
	\begin{lemma}
		The differential denominator $e\in k[t]$ of $\{\omega_1,\dots,\omega_n\}$ is squarefree.
	\end{lemma}
	
	\begin{proof}
		Let $p$ be an irreducible factor of $e$ and write $e=e_0p$. Then
		$p\omega_i'=\sum_{j=1}^n \frac{m_{i,j}}{e_0}\omega_j$. By Proposition~\ref{PROP:order}, $\nu_P(p\omega_i')>0$ for any place $P$ of $K$ lying above~$p$,
		i.e., $p\omega_i'$ is integral at~$p$. Since $\{\omega_1,\ldots,\omega_n \}$ is
		an integral basis, $\frac{m_{i,j}}{e_0}\in\mathcal{O}_p$. Then $p\nmid e_0$ by $\gcd(e,m_{1,1},\dots,m_{n,n})=1$. Hence, $e$ is squarefree.
	\end{proof}
	
	\begin{remark}
		The differential denominator $e$ may have special irreducible factors.
		For example, let $K = \mathbb{Q}(t,t')$ with $(t')^{3} = t$.
		Then $0$ is a special point of $X^{3} - t$, and hence $t$ is special.
		Since $t'' = (t')^{2}/(3t)$, the differential denominator of
		$\{1, t', t''\}$ is divisible by $t$.
		
	\end{remark}
	
	The above reduction does not introduce higher multiplicities.
	The remaining difficulty lies in the possible appearance of new special factors.
	In general, special poles of integrands are hard to handle.
	Fortunately, newly-introduced special factors can be removed by an
	additional modification as follows.
	
	Using notation in \eqref{EQ:HermiteG}, we assume $\gcd(e,v)=d$ and
	$e=e_1d$. Then $e_1$ is coprime with $v$ since $e$ is squarefree. By the
	extended Euclidean algorithm, we can find $\vec{a}^{\tau}=(a_1,\dots,a_n)^{\tau}\in k[t]^{n}$ such that $\deg_t a_i<\deg_t v$ and
	$v^{\mu-1}\vec{a}^{\tau}+\vec{r}^{\tau}=e_1\vec{b}^{\tau}$ for some
	$\vec{b}^{\tau}\in k[t]^{n}$. Set $g_a=\vec{a}^{\tau}\vec{\omega}$. Then
	\begin{equation}\label{EQ:addition}
		g'+g_a'=A+\bigl(\vec{a}^{\tau}\bigr)'\vec{\omega}+\frac{\bigl(\vec{r}^{\tau}+v^{\mu-1}\vec{a}^{\tau} \bigr)M\vec{\omega}}{e}=B+\frac{\vec{b}^{\tau}M\vec{\omega}}{d},
	\end{equation}
	where $B=A+\bigl(\vec{a}^{\tau} \bigr)'\vec{\omega}$. Since $v$ is normal, so is $d$.
	Then $B$ has no special factor in the denominator. Hence we
	eliminate the special poles of $g'$ by adding $g_a'$. Moreover, the denominator
	of $g_a'$ is a factor of~$e$, which is squarefree.
	
	Then by \eqref{EQ:reduction} and \eqref{EQ:addition}, the denominator of $f-(g+g_a)'$ is a product of normal irreducible factors, each of which has the multiplicity less than $\mu$ at~$v$. Repeating the reduction process until
	the integrand has a squarefree denominator, we arrive at:
	
	\begin{thm}\label{THM:HermiteReduction}
		Let $f\in K$. Then there exist $g\in K$, $f_0\in K$ with a normal denominator dividing~$e$, and $h=\sum_{i=1}^{n}\frac{h_i}{D_{*}}\omega_i$ where $D_*$ is normal and coprime with~$e$,
		$\deg h_i<\deg D_{*}$, such that
		\[
		\cN (f) = g'+f_0+h.
		\]
		Moreover, $h$ is unique, and $h=0$ if $f\in K'$.
	\end{thm}
	\begin{proof}
		By the preceding discussion, one can find $g\in K$ such that $w=\cN (f)-g'$ has a
		normal denominator. Write $w=\sum_{i=1}^{n}w_i\omega_i$ where $w_i\in k(t)$, then applying the
		extended Euclidean algorithm to the numerators of $w_i$ yields the desired
		decomposition $w=f_0+h$.
		
		It remains to prove the uniqueness of $h$ and verify $h=0$ if $f$ is a derivative in $K$. For proving the uniqueness and in-field integrability of $h$, it suffices to show that $h=0$ if either $h+f_0$ or $h+f_0+\cS (f)$ belongs to $K'$.
		
		Let $F\in K$ be such that $F'=h+f_0+\alpha$, where $\alpha=0$ or $\alpha=\cS (f)$. For any normal irreducible $p\in k[t]$, let $P$ be any place of $K$ lying above $p$ with ramification index
		$r_P$. Then $\nu_{P}(F)\ge 0.$ For, otherwise, $\nu_{P}(F')<-r_P$ by
		Proposition~\ref{PROP:order}, i.e., $\nu_{P}(pF')<0$, which would contradicts to the fact that the multiplicity of $p$ in the denominator of $F'$ is at most one.
		
		The conclusion $\nu_{P}(F)\ge 0$ implies $\cN (F)=0$. Then
		$F$ can be written $\frac{\vec{\rho}^{\tau}\vec{\omega}}{H}$, where $\vec{\rho}^{\tau}\in k[t]^{n}$,
		$\vec{\omega}=(\omega_1,\dots,\omega_n)^{\tau}$ and $H\in k[t]$ has only
		irreducible special factors. Then
		\[
		F'=\frac{\bigl(\vec{\rho}^{\tau}\bigr)'\vec{\omega}}{H}-\frac{H'\vec{\rho}^{\tau}\vec{\omega}}{H^2}+\frac{\vec{\rho}^{\tau}M\vec{\omega}}{eH}.
		\]
		Since $H'$ and all entries of $\bigl(\vec{\rho}^{\tau}\bigr)'$ are integral over $k[t]$, the denominator of $F'$ has no irreducible factor, which is normal and
		coprime with $e$. Hence, $h=0$.
	\end{proof}
	
	We call the element $h$ in Theorem~\ref{THM:HermiteReduction} the \emph{Hermite remainder} of~$f$ (w.r.t. $\{\omega_1,\dots,\omega_n\}$). The Hermite reduction described above naturally translates into an algorithm for computing $f_0$, $g$ and $h$ in Theorem~\ref{THM:HermiteReduction}. The result depends on the choice integral bases.
	%\begin{remark} The additional step for eliminating special factors oftroduced by differentiating
	%	$\omega_i$ essentially transforms the special poles to the poles at
	%	infinity. An alternative way of reduction is to skip this additional step if
	%	the specialness in the denominator is acceptable.
	%\end{remark}
	\begin{example}\label{EX:normal}
		
		Let $k=\mathbb{C}(z)$ equipped with $'=d/dz$ and $k(t,t')$ be a Weierstrass-like extension over $k$, where $(t')^2=4t^3-g_2t-g_3$, $g_2 = 0$ and $g_3 = -4$. We compute the integral of the function
		\[
		f = \frac{(t^2-t-1)t'-4+(2z+2)t^4+(4z+2)t^3-4zt^2-4t}{(t+1)t^2}.
		\]
		Let $\{1,t'\}$ be the chosen integral basis of $k(t,t')$. One can check that $0$ is not a special point but $-1$ is. Hence $t$ is normal and $t+1$ is special. By applying the extended Euclidean algorithm, we obtain the canonical representation of $f$ w.r.t. $\{1,t'\}$: $f=\cN (f) + \cS (f)$,
		where
		\[
		\cN (f) = \frac{-4-t'}{t^2},
		\qquad
		\cS (f) = \frac{2(z+1)t^2+2(2z+1)t-4z+t'}{t+1}.
		\]
		
		To reduce the multiplicities of $t$ in the denominator of $\cN (f)$, we set $\psi_1 = t^2 (t^{-1})' = -t'$ and $\psi_2 = t^2 (t^{-1}t')' = 2t^3-4$.
		By Lemma~\ref{LEM: localbasis}, $\{\psi_1,\psi_2\}$ is a local
		integral basis at $t$.
		We can find $c_1=1$ and $c_2=-\frac{4}{2t^3-4}$ such that the numerator of $\cN (f)$ is $c_1\psi_1 + c_2\psi_2$. Then we can compute $r_1=1$ and $r_2=1$ such that $\deg_t r_i< \deg_t t$ and $r_i\equiv c_i \mod t$. Set $g=\frac{r_1+r_2t'}{t}$, we have
		$\cN (f) = g'-2t$.
		Hence the Hermite remainder of $f$ is $0$.
		
		The special part $\cS (f)$ will be handled in the next section.
	\end{example}
	
	\section{Special and Polynomial Reductions}\label{SECT:CT}
	In this section, let $K=k(t,t')$ be a Weierstrass-like extension of $k$ and $m\in k[t,X]$ be the monic minimal polynomial of $t'$. We assume that $m=X^2-q$, where $q\in k[t]$ is squarefree and $\deg_tq\ge 3$. In such $K$, one can model extensions generated by transcendental Weierstrassian elements over~$k$.
	
	For $f\in K$, Theorem~\ref{THM:HermiteReduction}
	decomposes its normal part into the sum of a Hermite remainder $h$,
	an in-field integrable part $g'$, and an obstacle $f_0$ that admits no
	in-field integrability.
	Since such an obstacle may occur, Hermite reduction serves merely as a
	preprocessor for the normal part.
	Moreover, the special part of $f$ is not addressed.

	The goal of this section is to control these untreated parts by two further reductions, which will be called \emph{special reduction} and \emph{polynomial reduction}, respectively.
	
	At first, we reduce the special part. In general, it is difficult to determine all special points of $m$, which is equivalent to finding all algebraic solutions of a first-order differential equation. In order to circumvent this difficulty, we make a technical assumption throughout this section:
	
	\begin{hypo}\label{HYPO:constantspecial}
		Every special point of $m$ is a constant in $\overline{k}$.
	\end{hypo}
	
	By the above hypothesis, $q(\beta)=0$ if $\beta$ is a special point of $m$. Thus the special points of $m$ are constant roots of~$q$.
	
	\begin{remark}\label{rmk}
		The hypothesis holds when $k$ is a Liouvillian extension of $C_k$ and $q\in C_k[t]$ by~\cite[ Proposition 3.2 ]{Srinivasan2017}.
	\end{remark}

	\begin{lemma}
		Let $\beta\in  \overline{k}$ and $\beta'=0$. If $p\in k[t]$ is the monic minimal polynomial of $\beta$, then $p\in C_k[t]$. In particular, if $p\in k[t]$ is irreducible and special, then $p$ is a factor of $q$ with constant coefficients.
	\end{lemma}
	
	\begin{proof}
		Differentiating both sides of $p(\beta)=0$, we see that $\kappa(p)(\beta)+\beta'\partial_t(p)(\beta)=0$. Hence $\kappa(p)(\beta)=0$, i.e., $p\mid \kappa(p)$. Since $p$ is monic, we have that $\deg_t\kappa(p)< \deg_tp$. So $\kappa(p)=0$, i.e., all coefficients of $p$ are constants.
	\end{proof}
	
	\begin{lemma}\label{LEM:wpspecial}
		Let $p\in k[t]$ be irreducible and special. Then there exists exactly one place $P$ of $K$ lying above $p$. In particular, $\nu_{P}(p)=2$, $\nu_{P}(p')=1$ and $\nu_{P}(t')=1$.
	\end{lemma}
	
	\begin{proof}
		Since $p$ is irreducible and special, we have that $p\mid q$ and $p\in C_k[t]$. Let $P$ be a special place lying above $p$ with ramification index $r_P$ and set $v=\nu_{P}(t')$. As $p'=t'\partial_t(p)$ and $\gcd\bigl(p,\partial_t(p)\bigr)=1$, we have $\nu_{P}(p')=v$. Since $(t')^2=q$ and $q$ is squarefree, it follows that $2v=r_P$. Then $r_P\ge 2$. By~\cite[page 52, Theorem 1]{chevalley1951}, $P$ is the only place lying above $p$ and $r_P=[K:k(t)]=2$. Hence, $v=1$.
	\end{proof}
	
	All special places of $K$ are of the zeroth kind by the above lemma. Then $C_k$ is the field of constants of $K$ by Corollary~\ref{COR:constant}.
	
	By~\cite[page 31]{Trager1984}, $\{1,t' \}$ is an integral basis of $K$. Write $q=q_Nq_S$, where $q_N\in k[t]$ is normal, $q_S\in k[t]$ is monic and all its irreducible factors are special. Then $q_S\in C_k[t]$, and
	\[
	t''=\frac{\kappa(q_N)}{2q_N}t'+\frac{\partial_t(q)}{2}.
	\]
	Hence $q_N$ is the differential denominator of $\{1,t'\}$.
	Denote by $I_K$ the set of elements in $K$ that are integral over $k[t]$. Under the basis $\{1,t'\}$, Hermite reduction simplifies the normal parts of integrands.
	
	We now describe the special reduction, which decreases the multiplicity of factors of the denominator of special parts modulo derivatives. Assume that $f\in K$ with the special part $\cS (f)=\frac{A+Bt'}{D}$, where $A,B,D\in k[t]$, $D$ is monic and $\gcd(A,B,D)=1$. Assume that $D=uv^{\mu}$ where $\mu>0$, $v$ is squarefree and coprime with $u$, and factors of $u$ have multiplicity less than $\mu$. By Hypothesis~\ref{HYPO:constantspecial}, $u,v\in C_k[t]$ and $v$ divides $q$. Set $q_v=q/v\in k[t]$. For $a,b\in k[t]$ and $\lambda\in \mathbb{N}\setminus\{0\}$, a direct calculation shows that
	\begin{equation}\label{EQ:specialred1}
		\left( \frac{a}{v^\lambda}\right)'=\frac{\kappa(a)+\partial_t(a)t'}{v^\lambda}-\frac{\lambda a\partial_t(v)t'}{v^{\lambda+1}} 
	\end{equation}
	and
	\begin{equation}\label{EQ:specialred2}
	\left(\frac{bt'}{v^\lambda}\right)'=\frac{b_1t'}{q_Nv^\lambda}+\frac{b_2}{v^{\lambda-1}}+\frac{(1-2\lambda)bq_v\partial_t(v)}{2v^\lambda}
	\end{equation}
	for some $b_1,b_2\in k[t].$ It proceeds as follows:
	
	\begin{itemize}
		\item[(i)] If $\mu\ge 1$, we compute $b\in k[t]$ such that $\deg_t b< \deg_t v$ and 
		\[
		(1-2\mu)ubq_v\partial_t(v)\equiv 2A \mod v.
		\] Such $b$ can be found since $q_v$, $\partial_t(v)$ and $u$ are coprime with~$v$. It follows from \eqref{EQ:specialred2} that
		\begin{equation}\label{EQ:specialred}
			\tilde{f}:=\cS (f)-\left(\frac{bt'}{v^\mu}\right)'=\frac{\tilde{A}}{uv^{\mu-1}}+\frac{\tilde{B}t'}{uv^\mu}+R
		\end{equation}
		for some $\tilde{A},\tilde{B}\in k[t]$ and $R\in \frac{I_K}{q_N}$.
		\item[(ii)] If $\mu$ in \eqref{EQ:specialred} is greater than or equal to $2$, then we compute $a\in k[t]$ such that $\deg_t a< \deg_t v$ and 
		\[
		(\mu-1)ua\partial_t(v)\equiv\tilde{B}\mod v.
		\] Such $a$ can be found since $u,\partial_t(v)$ are coprime with $v$. By \eqref{EQ:specialred1}, $\tilde{f}+\left( \frac{a}{v^{\mu-1}}\right)'$ has a denominator, in which the multiplicity of $v$ is at most $\mu-1$.
	\end{itemize}
	Repeating (i) and (ii) to $\cS (f)$, we have
	
	\begin{theorem}\label{THM:SpecialReduction}
		Let $f\in K$. Then there exist $g\in K$, $f_1\in \frac{I_K}{q_N}$ and $s=\frac{\theta}{\gamma}t'$, where $\gamma\in k[t]$ is squarefree with only irreducible special factors, $\theta\in k[t]$ and $\deg_t\theta<\deg_t\gamma$, such that
		\[
		\cS (f)=g'+f_1+s.
		\]
		Moreover, $s$ is unique and $s=0$ if $f\in K'$.
	\end{theorem}
	\begin{proof}
		By (i) and (ii) given above, there exists $g\in K$ such that $\cS (f)-g'=\frac{\theta_0}{\gamma}t'+r$, where $\gamma,\theta_0\in k[t]$, $r\in \frac{I_K}{q_N}$, and $\gamma$ is squarefree with merely special factors. Dividing $\theta_0$ by $\gamma$ yields the desired $\theta$ and $f_1$.
		
		Now we prove the uniqueness and in-field integrability of $s$. Similar to the proof of Theorem~\ref{THM:HermiteReduction}, it suffices to prove that $s=0$ if $f_1+s+\alpha\in K'$, where $\alpha=0$ or $\alpha=\cN (f)$.
		
		Assume $s\ne 0$ and let $F\in K$ satisfy $F'=f_1+s+\alpha$. Let $p$ be an irreducible factor of $\gamma$ and $P$ be the place lying above $p$. By Lemma~\ref{LEM:wpspecial}, $\nu_P(p)=2$ and $\nu_{P}(p')=\nu_{P}(t')=1$.
		
		We claim that $\nu_{P}(F)\ge 0$. Otherwise,
		$\nu_{P}(F')=\nu_{P}(F)-1\le -2$
		by Proposition~\ref{PROP:order}. On the other hand, $\gamma$ is squarefree, and thus, $\nu_{P}(s)\ge \nu_{P}(\frac{t'}{\gamma})=-1$. Accordingly, $\nu_{P}(F')\ge -1$ since $f_1$ and $z$ are integral at $p$,  a contradiction. The claim holds.
		
		Again by Proposition~\ref{PROP:order}, $\nu_{P}(F)\ge 0$ implies $\nu_{P}(F')\ge 0$. Since $P$ is the only place lying above $p$, $F'$ is integral at $p$, which contradicts the fact that $\{1,t'\}$ is an integral basis and $F'=f_1+\frac{\theta}{\gamma}t'+z$.
	\end{proof}
	
	For a given $f\in K$, the special reduction described above computes $f_1,g$ and $s$ in Theorem~\ref{THM:SpecialReduction}. We call $s$ the \emph{special remainder} of~$f$ (w.r.t. $\{1,t'\}$). Now we show how to use special reduction to integrate the $\cS (f)$ in Example~\ref{EX:normal}.
	\begin{example}
		
		We have $q=4(t^3+1)$. Then $q_N=4$ and $q_S=t^3+1$. Recall that $\cS (f)=\frac{A+Bt'}{D}$ where $A=2(z+1)t^2+2(2z+1)t-4z$, $B=1$ and $D=t+1$.
		Then $u=1$, $v=t+1$, $\mu=1$ and $q_v = 4(t^2-t+1)$.
		One can find that $b=z$ such that $\deg_t b < \deg_t v$ and $(1-2\mu)\,u b\, q_v\, \partial_t(v) \equiv 2A \mod{v}$.
		Then $\tilde{A}=0$, $\tilde{B}=0$ and $R=2t$ in~\eqref{EQ:specialred}.
		Therefore, $\cS (f) =\left( \frac{zt'}{t+1} \right)'+ 2t  $. The special reminder of $f$ is $0$.
		
		Combine with the result in Example~\ref{EX:normal}, we see that the obstacles in $\cN (f)$ and $\cS (f)$ are canceled with each other, hence $f\in K'$, i.e.,
		\[\int f \, dz = \frac{1+t'}{t} + \frac{zt'}{t+1}.\]
		
	\end{example}
	
	By Theorems~\ref{THM:HermiteReduction} and \ref{THM:SpecialReduction},
	each element of $K$ is decomposed as the sum of its Hermite remainder, special remainder and an element in $\frac{I_K}{q_N}$. To control the poles at infinity, we develop the polynomial reduction to simplify the elements of $\frac{I_K}{q_N}$.
	
	For $f\in \frac{I_K}{q_N}$, we can write  $f=\frac{a+bt'}{q_N}+w+rt'$, where $a, b, w, r \in k[t]$ with $\deg_t a$ and  $\deg_t b$ are less than $\deg_t q_N$. Then 
	\[
	\cN (f)=\frac{a+bt'}{q_N}\,\, \text{and} \,\, \cS (f)= w+rt'.
	\] Moreover, $\cS (f)$ is integral over $k[t]$. Write $r=r_nt^n+\dots+r_0$, and set $\bUpsilon(f):=\frac{r_n}{n+1}t^{n+1}+\dots+r_0t\in k[t]$. Then $\partial_t\bigl(\bUpsilon(f)\bigr)=r$. By a direct calculation,
	\[
	f-\bUpsilon(f)'=f-\kappa\bigl(\bUpsilon(f)\bigr)-rt'=\cN (f)+w-\kappa\bigl(\bUpsilon(f)\bigr).
	\]
	So the special part of $f-\bUpsilon(f)'$ is $w-\kappa\bigl(\bUpsilon(f)\bigr)$, denoted by $\cS^*(f)$, which belongs to $k[t]$. Moreover, $\cS^*$ is a $k$-linear operator on $\frac{I_K}{q_N}$.
	\begin{lemma}\label{LEM:polyred}
		For any $\delta\in k$ and $\lambda\in\mathbb{N}$, we have $(\delta t^\lambda t' )'\in \frac{I_K}{q_N}$ and $\cS^*\bigl((\delta t^\lambda t' )'\bigr)$ is of degree $\lambda-1+\deg_t q$ with leading coefficient $l(\lambda)\delta$, where 
\[l(\lambda)=\left(\lambda+\frac{\deg_tq}{2}\right)\lc_t(q).\]
	\end{lemma}
	\begin{proof}
		A direct calculation shows that
		\[
		\bigl(\delta t^\lambda t'\bigl)'=\underbrace{\lambda\delta qt^{\lambda-1}+\frac{\delta \partial_t(q)t^{\lambda}}{2}}_{\Gamma_1}+\underbrace{\left(\delta't^\lambda+\delta\frac{\kappa(q_N)}{2q_N}t^\lambda\right)}_{\Gamma_2}t'.
		\]
		Hence $(\delta t^\lambda t' )'\in \frac{I_K}{q_N}$. Furthermore, the polynomial $\Gamma_1\in k[t]$ is of degree $\lambda-1+\deg_tq$ with leading coefficient $l(\lambda)\delta$, $\Gamma_2\in k(t)$ and $\cS (\Gamma_2)$ is either $0$ or of degree $\lambda$. Set $R:=\bUpsilon\bigl((\delta t^\lambda t')'\bigr)$. Then $\partial_t(R)=\cS (\Gamma_2)$ and $\cS^*\bigl((\delta t^\lambda t' )'\bigr)=\Gamma_1-\kappa(R)$. In particular, $\deg_t R\le \lambda+1$.
		From the assumption that $\deg_t q\ge 3$, we have
		\[
		\deg_t\kappa(R)\le \lambda+1 < \lambda-1+\deg_t q=\deg_t \Gamma_1.
		\]
		Thus  $\cS^*\bigl((\delta t^\lambda t' )'\bigr)$ has the same leading term as $\Gamma_1$.
	\end{proof}
	
	Given $f\in \frac{I_K}{q_N}$, we let $d=\deg_t \cS^*(f)$ and $\varepsilon=\lc_t\bigl(\cS^*(f) \bigr)$. If $d\ge \deg_t q-1$, then we take $\lambda$ and $\delta$ in the above lemma as $d-\deg_t q+1$ and $\varepsilon/l(\lambda)$, respectively. The same lemma implies the leading term of $\cS^*\bigl((\delta t^\lambda t' )'\bigr)$ is equal $\cS^*(f)$.
	
	Let $\tilde{f}=f-(\delta t^\lambda t')'$. Then the degree of $\cS^*\bigl(\tilde{f}\bigr)$ is less than $d$. Repeating this process to $\tilde{f}$ until $\deg_t \cS^*(\tilde{f})<\deg_t q-1$, we find $b\in k[t]$ such that $\cS^*(f-b')$ has degree less than $\deg_t q-1$. With this degree-decreasing process, we have
	\begin{theorem}\label{THM:PolyReduction}
		Let $f\in \frac{I_K}{q_N}$. Then there exist $g\in K$, $f_2\in \frac{I_K}{q_N}$ with $\cS (f_2)=0$ and $\eta\in k[t]$ with $\deg_t \eta <\deg_t q - 1$, such that
		\[
		f=g'+f_2+\eta.
		\]
		Moreover, $f_2$ is unique, and $f\in K'$ if and only if $f_2=0$ and $\eta\in k'$.
	\end{theorem}
	The proof of Theorem~\ref{THM:PolyReduction} is based on the next lemma.
	\begin{lemma}\label{LEM:inforder}
		\begin{itemize}
			\item[(i)] Let $P$ be an infinite place of $K$ and let $f=a+bt'$ be such that $a,b\in k(t)$ are proper fractions. Then $\nu_{P}(f)\ge \nu_{P}(t')-\nu_{P}(t)$.
			\item[(ii)] Let $F\in K$ be integral over $k[t]$ and polynomial $\eta\in k[t]$ with $\deg_t \eta<\deg_t q-1$. If $\nu_{P}(F'+\eta)\ge \nu_{P}(t')-\nu_{P}(t)$
			for any infinite place $P$ of $K$, then $F\in k$.
		\end{itemize}
		
	\end{lemma}
	\begin{proof}
		(i) Let $r_P$ be the ramification index of $P$. For any polynomial $w\in k[t]$, we have $\nu_{P}(w)=-r_P\deg_t w$ since $\nu_{P}(t)=-r_P<0$. Then $\nu_{P}(t')<0$ by $(t')^2=q$. As $a$ and $b$ are proper,
		\[
		\nu_{P}(a)\ge r_P>r_P+\nu_{P}(t') \,\,\text{and}\,\,\nu_{P}(bt')\ge r_P+\nu_{P}(t').
		\]
		
		(ii) Since $\{1,t' \}$ is an integral basis, write $F=A+Bt'$ with $A,B\in k[t]$. Then $F'=A_0+B_0t'+\cN (F')$, where
		\[
		A_0=\kappa(A)+\partial_t(B)q+\frac{1}{2}B\partial_t(q),\,\, B_0=\partial_t(A)+\kappa(B)+S\left(\frac{\kappa(q_N)B}{2q_N} \right),
		\]
		and $A_0$, $B_0 \in k[t]$. By (i), $\nu_{P}\bigl(\cN (F')\bigr)\ge \nu_{P}(t')-\nu_{P}(t)$. Since
		\[
		A_0+\eta+B_0t'=F'+\eta-\cN (F'),
		\] we have that $\nu_{P}(A_0+\eta+B_0t')$ is no less than $\nu_{P}(t')-\nu_{P}(t)$ for any infinite place $P$. Set
		\[
		C:=\frac{t}{t'}(A_0+\eta+B_0t')=\frac{(A_0+\eta)t}{q}t'+B_0t.
		\]
		Then $C$ is integral at $t^{-1}$. By \cite[page~30, Proposition]{Trager1984}, $\{1,t'\}$ is normal at $t^{-1}$. It follows from \cite[Lemma~2]{ChenKauersKoutschan2016} that
		$\frac{(A_0+\eta)t}{q}t'$ and $B_0t$ are integral at $t^{-1}$. Accordingly, $B_0=0$.
		
		We claim that $B=0$. Otherwise, $B_0=0$ and $B\ne 0$ imply that $\deg_t \partial_t(A)\le \deg_t B$, i.e., $\deg_t A\le \deg_t B+1$. Since $\deg_t q\ge 3$, we have
		$\deg_t \kappa(A)\le \deg_t A\le \deg_t B+1 < \deg_t B+\deg_t q-1$.
		Note that
		%$\partial_t(B)q$ and $\frac{1}{2}B\partial_t(q)$ are of degree $\deg_t B+\deg_t q-1$ with leading coefficients differing by a positive rational number,
		the degree of $\partial_t(B)q+\frac{1}{2}B\partial_t(q)$ is equal to $\deg_t B+\deg_t q-1$, which is greater than $\deg_t\kappa(A)$. Hence 
		\[
		\deg_t A_0=\deg_t B+\deg_t q-1>\deg_t \eta.
		\]
		It follows that
		\[
		\deg_t (A_0+\eta)t=\deg_t A_0t=\deg_t B+\deg_t q \ge \deg_t q,
		\]
		hence $\frac{(A_0+\eta)t}{q}$ admits nonpositive order at $t^{-1}$. As $t'$ is not integral at $t^{-1}$, neither is $\frac{(A_0+\eta)t}{q}t'$, a contradiction. The claim holds.
		
		Consequently, $A\in k$ by $B_0=0$. It follows that $F\in k$ .
	\end{proof}
	\begin{proof}[Proof of Theorem~\ref{THM:PolyReduction}]
		Let $b\in k[t]$ be such that the degree of $\cS^*(f-b')$ is less than $\deg_t q-1$. By the definition of $\cS^*$, we have
		\[
		f-b'-\bigl(\bUpsilon(f-b')\bigr)'=\cN (f-b')+\cS^*(f-b').
		\]
		Setting $g=b+\bUpsilon(f-b')$, $f_2=\cN (f-b')$ and $\eta=\cS^*(f-b')$ gives the desired decomposition.
		
		For the uniqueness of $f_2$ and in-field integrablity condition for~$f$, it suffices to prove that $f_2+\eta\in K'$ implies $f_2=0$ and $\eta\in k'$. Assume $F'=f_2+\eta$ for some $F\in K$. By an order comparison similar to those in the proofs of Theorems~\ref{THM:HermiteReduction} and~\ref{THM:SpecialReduction}, $F$ is integral over $k[t]$. Let $P$ be an infinite place of $K$ with ramification index $r_P$. By Lemma~\ref{LEM:inforder}~(i), $\nu_{P}(f_2)\ge \nu_{P}(t')-\nu_{P}(t)$. Hence $\nu_{P}(F'-\eta)\ge \nu_{P}(t')-\nu_{P}(t)$. By Lemma~\ref{LEM:inforder} (ii), $F\in k$. Thus $f_2=0$ and $\eta=F'\in k'$.
	\end{proof}
	 The process of polynomial reduction naturally translates into an algorithm for computing $f_2,g$ and $\eta$ in Theorem~\ref{THM:PolyReduction}. Combining the Hermite, special and polynomial reductions, we have the following theorem.
	\begin{thm}\label{THM:WholeReduction}
       For $f\in K$, we let $h$ be the Hermite remainder and $s$ be the special remainder of $f$ w.r.t. $\{1,t'\}$ as in Theorems~\ref{THM:HermiteReduction} and \ref{THM:SpecialReduction}, respectively. Then there exists $g\in K$, a unique element $l\in \frac{I_K}{q_N}$ with no special part, and $\eta\in k[t]$ with $\deg_t \eta < \deg_t q -1$ such that

		\[
		f=g'+h+s+l+\eta.
		\]
		Moreover, $f\in K'$ if and only if $h$, $s$ and $l$ are all zero and $\eta\in k'$.
	\end{thm}
	\begin{proof}
		The existence of $g,l$ and $\eta$ follows from Theorems~\ref{THM:HermiteReduction}, \ref{THM:SpecialReduction} and~\ref{THM:PolyReduction}. If $f\in K'$ or $f=0$, then $l+\eta\in K'$. Hence $l=0$ and $\eta\in k'$ by Theorem~\ref{THM:PolyReduction}.
	\end{proof}
	Although $\eta$ in Theorem~\ref{THM:WholeReduction} is not unique, it is determined up to an element in $k'$ additively. Hence the positive degree terms of $\eta$ are unique. We call such $\eta$ a \emph{polynomial remainder} of~$f$ (w.r.t. $\{1,t' \}$).
	
	%\begin{remark}
	%An easy refinement to the term of degree $0$ of $\eta$ can make the reduction complete if a complete reduction in $k$ is available.
	%\end{remark}
	
	Theorem~\ref{THM:WholeReduction} is not only a criterion for in-field integrablity in Weierstrass-like extensions, but also leads to a necessary condition for elementary integrablity.
	
	\begin{corollary}\label{COR:elementaryint}
		Assume that $C_k$ is algebraically closed. Let $f\in K$ and $\eta$ be a polynomial remainder of $f$. If $f$ has an elementary integral over $K$, then $\deg_t \eta\le \frac{\deg_t q}{2}-1$.
	\end{corollary}
	\begin{proof}
		  Write $f=g'+h+s+l+\eta$ as in Theorem~\ref{THM:WholeReduction} and set $R=f-g'$. If $f$ has an elementary integral over $K$, so does $R$. By~\cite[Thm 5.5.2]{BronsteinBook}, there exist $F\in K$, $c_1,\dots,c_n\in C_k$ and $u_1,\dots,u_n\in K\setminus\{0\}$ such that
		\[
		R=F'+\sum_{i=1}^{n}c_i\frac{u_i'}{u_i}.
		\]
		
		Let $P$ be a place of $K$ with ramification index $r_P$. If $P$ is normal, then $\nu_{P}(\frac{u_i'}{u_i})\ge -r_P$ by Proposition~\ref{PROP:order}. Similar to the proof of Theorem~\ref{THM:HermiteReduction}, one can show that $F$ has no normal part. If $P$ is special, then $\nu_{P}(\frac{u_i'}{u_i})\ge -1$ by Proposition~\ref{PROP:order} and Lemma~\ref{LEM:wpspecial}. Similar to the proof of Theorem~\ref{THM:SpecialReduction}, one can show that $F$ is integral over $k[t]$.
		
		If $P$ is an infinite place, then $\nu_{P}(t')=-\frac{r_P\deg_t q}{2}$ by $(t')^2=q$, which implies $\nu_{P}(t')+r_P=\frac{r_P(2-\deg_t q)}{2}<0$. Then
		\[
		\nu_{P}\left(\frac{u_i'}{u_i}\right)\ge \nu_{P}(t')+r_P=\nu_{P}(t')-\nu_{P}(t)
		\] by Proposition~\ref{PROP:orderinf}. Moreover, $\nu_{P}(h)$, $\nu(s)$ and $\nu_{P}(l)$ are all greater than or equal to $\nu_{P}(t')-\nu_{P}(t)$  by Lemma~\ref{LEM:inforder} (i). Hence $\nu_{P}(F'-\eta)\ge \nu_{P}(t')-\nu_{P}(t)$. Then $F\in k$ by of Lemma~\ref{LEM:inforder} (ii).
		
		If $\deg_t \eta>0$, then
		\[
		\nu_{P}(t')+r_P=\nu_{P}(t')-\nu_{P}(t)\le \nu_{P}(F'-\eta)=\nu_{P}(\eta).
		\]
		Hence $\deg_t \eta=-\frac{\nu_{P}(\eta)}{r_P} \le -\frac{\nu_{P}(t')}{r_P}-1=\frac{\deg_t q}{2}-1$.
	\end{proof}
	
	\section{The Appetizer Revisited}\label{SECT:revisit}
	
	We now apply the Hermite reduction and Weierstrass reduction to evaluate the integrals $I_n(z) :=  \int \wp(z)^n \, dz$ with $n\in \bN$ in Section~\ref{SECT:start}. Let $k=\mathbb{C}(z)$ be the field of
	rational functions equipped with the derivation $':=d/dz$. Then the field of
	constants of $k$ is $\mathbb{C}$. Let $K = k(t,t')$ be a Weierstrass-like extension and $m=X^2-q\in k[t,X]$ be the monic minimal polynomial of $t'$, where $q=4t^3 - g_2 t - g_3$ with $g_2,g_3\in\mathbb{C}$ and $27g_3^2 - g_2^3 \neq 0$. In this sense, $t$ satisfies the same differential equation as the Weierstrass-$\wp$ function. Then $\{1, \; t'\}$ is an integral basis of $K$.
	
	Hypothesis~\ref{HYPO:constantspecial} holds for our setting by Remark~\ref{rmk}, i.e., any special point of $m$ is a constant in ${\mathbb{C}}$. Then $t$ itself is a polynomial remainder. Since $\frac{\deg_t q}{2}-1 < 1$, $t$ has no elementary integral over $K$ by Corollary~\ref{COR:elementaryint}. As in Section~\ref{SECT:start}, $\zeta(z)$ stands for the integral of $t$.
	
	Note that $t^n$ and $(t^nt')'$ lie in $k[t]$ for $n\in \mathbb{N}$. It follows that $\cS^*(t^n)=t^n$ and $\cS^*\bigl( (t^nt')\bigr)=(t^nt')'$.
	Applying the polynomial reduction to $t^2$ yields that
	\begin{equation*}
		t^2=\left(\frac{1}{6}t'\right)'+\frac{g_2}{12}.
	\end{equation*}
	Hence a polynomial remainder of $t^2$ is $\frac{g_2}{12}$, $t^2\in K'$, i.e.,
	\[
	\int t^2 \, dz = \frac{1}{6} t' + \frac{g_2}{12} z .
	\]
	Applying the polynomial reduction to $t^3$, we find that
	\[
	t^3 = \left(\frac{1}{10}tt'\right)' +\frac{3g_2}{20}t + \frac{g_3}{10}.
	\]
	Then $t^3$ has a polynomial remainder $\frac{3g_2}{20}t+\frac{g_3}{10}$. So $t^3$ has no elementary integral over $K$ by Corollary~\ref{COR:elementaryint}. Using $\zeta(z)$, we can represent the integral as
	\[\int t^3 \, dz = \frac{1}{10} tt' -\frac{3g_2}{20} \zeta + \frac{g_3}{10} z .\]
	Similarly, applying the polynomial reduction to $t^4$ yields that
	\[
	t^4 = \left(\frac{1}{14}t^2t'\right)'  -\frac{g_3}{7}t + \frac{5g_2^2}{336}
	\]
	Then $-\frac{g_3}{7}t+\frac{5g_2^2}{336}$ is a polynomial remainder of $t^4$, which implies $t^4$ has no elementary integral over~$k$. With the help of $\zeta(z)$, the integral of $t^4$ is represented as
	\[\int t^4 \, dz = \frac{1}{14}t^2t' + \frac{5g_2}{168} t' + \frac{5g_2^2}{336}z-\frac{g_3}{7}\zeta.
	\]
	In fact, for any $n\in \bN$, $t^n$~admits the following decomposition:
	\[
	t^n = \left(\frac{1}{4n-2}t^{n-2}t'\right)' + \frac{(n-2)g_3}{4n-2}t^{n-3} + \frac{(2n-3)g_2}{8n-4} t^{n-2}.
	\]
	Substituting $t=\wp(z)$ into the identity and integrating w.r.t. $z$ yields the claimed recurrence for $J_n(z)$ in Section~\ref{SECT:start}.
	%\section{Conclusion}
	%\newpage

%\balance
\bibliographystyle{plain}

	%\bibliography{integral}

\begin{thebibliography}{10}

\bibitem{Boettner2010}
Stefan~T. Boettner.
\newblock {\em Mixed Transcendental and Algebraic Extensions for the
  Risch-Norman Algorithm}.
\newblock Phd thesis, Tulane University, New Orleans, USA, 2010.

\bibitem{bostan10b}
Alin Bostan, Shaoshi Chen, Fr\'{e}d\'{e}ric Chyzak, and Ziming Li.
\newblock Complexity of creative telescoping for bivariate rational functions.
\newblock In {\em Proceedings of the 35th International Symposium on Symbolic
  and Algebraic Computation}, page 203–210. ACM, 2010.

\bibitem{BCCLX2013a}
Alin Bostan, Shaoshi Chen, Fr\'{e}d\'{e}ric Chyzak, Ziming Li, and Guoce Xin.
\newblock Hermite reduction and creative telescoping for hyperexponential
  functions.
\newblock In {\em Proceedings of the 38th International Symposium on Symbolic
  and Algebraic Computation}, page 77–84. ACM, 2013.

\bibitem{BCLS2018}
Alin Bostan, Fr\'{e}d\'{e}ric Chyzak, Pierre Lairez, and Bruno Salvy.
\newblock Generalized hermite reduction, creative telescoping and definite
  integration of {D}-finite functions.
\newblock In {\em Proceedings of the 43th International Symposium on Symbolic
  and Algebraic Computation}, page 95–102. ACM, 2018.

\bibitem{Bronstein1990}
Manuel Bronstein.
\newblock Integration of elementary functions.
\newblock {\em J. Symbolic Comput.}, 9(2):117--173, 1990.

\bibitem{BronsteinBook}
Manuel Bronstein.
\newblock {\em {S}ymbolic {I}ntegration {I}: {T}ranscendental {F}unctions}.
\newblock Springer-Verlag, Berlin, 2005.

\bibitem{Bronstein2007}
Manuel Bronstein.
\newblock Structure theorems for parallel integration.
\newblock {\em J. Symbolic Comput.}, 42(7):757--769, 2007.

\bibitem{chen21b}
Shaoshi Chen, Lixin Du, and Manuel Kauers.
\newblock Lazy {H}ermite reduction and creative telescoping for algebraic
  functions.
\newblock In {\em {P}roceedings of the 46th {I}nternational {S}ymposium on
  {S}ymbolic and {A}lgebraic {C}omputation}, pages 75--82. ACM, 2021.

\bibitem{CDK2023}
Shaoshi Chen, Lixin Du, and Manuel Kauers.
\newblock Hermite reduction for {D}-finite functions via integral bases.
\newblock In {\em Proceedings of the 48th International Symposium on Symbolic
  and Algebraic Computation}, page 155–163. ACM, 2023.

\bibitem{ChenKauersKoutschan2016}
Shaoshi Chen, Manuel Kauers, and Christoph Koutschan.
\newblock Reduction-based creative telescoping for algebraic functions.
\newblock In {\em Proceedings of the 41th ACM International Symposium on
  Symbolic and Algebraic Computation}, page 175–182. ACM, 2016.

\bibitem{Chen2018JSC}
Shaoshi Chen, Mark van Hoeij, Manuel Kauers, and Christoph Koutschan.
\newblock Reduction-based creative telescoping for fuchsian {D}-finite
  functions.
\newblock {\em J. Symbolic Comput.}, 85:108--127, 2018.

\bibitem{chevalley1951}
Claude Chevalley.
\newblock {\em {Introduction to the {T}heory of {A}lgebraic {F}unctions of
  {O}ne {V}ariable}}.
\newblock American Mathematical {S}urveys, 1951.

\bibitem{ChyzakSalvy1998}
Fr{\'e}d{\'e}ric Chyzak and Bruno Salvy.
\newblock Non-commutative elimination in {O}re algebras proves multivariate
  identities.
\newblock {\em J. Symbolic Comput.}, 26(2):187--227, 1998.

\bibitem{DGLL2025}
Hao Du, Yiman Gao, Wenqiao Li, and Ziming Li.
\newblock Complete reduction for derivatives in a primitive tower.
\newblock In {\em Proceedings of the 50th International Symposium on Symbolic
  and Algebraic Computation}, page 42–51. ACM, 2025.

\bibitem{DHL2018}
Hao Du, Hui Huang, and Ziming Li.
\newblock A {$q$}-analogue of the modified {A}bramov-{P}etkov\v{s}ek reduction.
\newblock In {\em Advances in computer algebra}, volume 226 of {\em Springer
  Proc. Math. Stat.}, pages 105--129. Springer, Cham, 2018.

\bibitem{Hermite1872}
Charles Hermite.
\newblock Sur l'int\'egration des fractions rationnelles.
\newblock {\em Annales Scientifiques de l'\'Ecole Normale Sup\'erieure.
  Deuxi\'eme S\'erie}, 1:215--218, 1872.

\bibitem{Kauers2023}
Manuel Kauers.
\newblock {\em {D}-{F}inite {F}unctions}.
\newblock Springer-Verlag, Cham, 2023.

\bibitem{Kolchin1953}
E.~R. Kolchin.
\newblock Galois theory of differential fields.
\newblock {\em American Journal of Mathematics}, 75(4):753--824, 1953.

\bibitem{koutschan10c}
Christoph Koutschan.
\newblock {HolonomicFunctions (User's Guide)}.
\newblock Technical Report 10-01, RISC Report Series, University of Linz,
  Austria, January 2010.

\bibitem{kubhakar23}
Partha Kumbhakar and Varadharaj~R. Srinivasan.
\newblock Liouville’s theorem on integration in finite terms for $d_\infty$,
  $sl_\infty$, and weierstrass field extensions.
\newblock {\em Archiv der Mathematik}, 121:371--383, 2023.

\bibitem{Ostrogradsky1845}
Mikhail~Vasil'evich Ostrogradski{\u\i}.
\newblock De l'int{\'e}gration des fractions rationnelles.
\newblock {\em Bull.\ de la classe physico-math{\'e}matique de l'Acad.\
  Imp{\'e}riale des Sciences de Saint-P{\'e}tersbourg}, 4:145--167, 286--300,
  1845.

\bibitem{pila22}
Jonathan Pila and Jacob Tsimerman.
\newblock Ax-{S}chanuel and exceptional integrability.
\newblock Technical Report 2202.04023, ArXiv, 2022.

\bibitem{raab22}
Clemens~G. Raab and Michael~F. Singer, editors.
\newblock {\em Integration in Finite Terms: Fundamental Sources}.
\newblock Texts and Monographs in Symbolic Computation. Springer, 2022.

\bibitem{Risch1969}
Robert~H. Risch.
\newblock The problem of integration in finite terms.
\newblock {\em Trans. Amer. Math. Soc.}, 139:167--189, 1969.

\bibitem{Risch1970}
Robert~H. Risch.
\newblock The solution of the problem of integration in finite terms.
\newblock {\em Bull. Amer. Math. Soc.}, 76:605--608, 1970.

\bibitem{Rybowicz91}
Marc Rybowicz.
\newblock An algorithms for computing integral bases of an algebraic function
  field.
\newblock In {\em Proceedings of the 16th International Symposium on Symbolic
  and Algebraic Computation}, page 157–166. ACM, 1991.

\bibitem{Srinivasan2017}
Varadharaj~Ravi Srinivasan.
\newblock Liouvillian solutions of first order nonlinear differential
  equations.
\newblock {\em Journal of Pure and Applied Algebra}, 221(2):411--421, 2017.

\bibitem{stichtenoth2009}
Henning Stichtenoth.
\newblock {\em Algebraic function fields and codes}.
\newblock Springer, 2009.

\bibitem{Trager1984}
Barry~M. Trager.
\newblock {\em {On the Integration of algebraic functions}}.
\newblock Phd thesis, { MIT, Computer Science}, 1984.

\bibitem{vanderHoeven21}
Joris van~der Hoeven.
\newblock Constructing reductions for creative telescoping: the general
  differentially finite case.
\newblock {\em Applicable Algebra in Engineering, Communication and Computing},
  32(5):575–602, nov 2021.

\bibitem{vanHoeij94}
Mark van Hoeij.
\newblock An algorithm for computing an integral basis in an algebraic function
  field.
\newblock {\em J. Symbolic Comput.}, 18(4):353--363, 1994.

\bibitem{WW1927}
E.~T. Whittaker and G.~N. Watson.
\newblock {\em A course of modern analysis}.
\newblock Cambridge Mathematical Library.

\bibitem{Zeilberger1990}
Doron Zeilberger.
\newblock A holonomic systems approach to special functions identities.
\newblock {\em J. Comput. Appl. Math.}, 32:321--368, 1990.

\end{thebibliography}

\appendix
	\section{Splitting Factorization}\label{APPEN:splitting}
	
	Let $K=k(t,t')$ be a Weierstrass-like extension over $k$ and let $m\in k[t,X]$ be the monic minimal polynomial of $t'$. For any $D\in k[t]$, there exist $D_N,D_S\in k[t]$ such that $D=D_N D_S$, where all irreducible factors of $D_N$ are normal and those of $D_S$ are special. This factorization is unique up to a nonzero multiplicative element in $k$, and is called the \emph{splitting factorization} of $D$ (w.r.t. $m$). In this section, an algorithm is presented to compute splitting factorization by gcd-computation and resultants.
	
	\begin{lemma}\label{LEM:splitting}
		Let $D\in k[t]$ be squarefree and set $D_0(t,y):=\kappa(D)+\partial_t(D)y$, where $y$ is an indeterminate. Let $R(t)\in k[t]$ be the Sylvester resultant of $D_0(t,y)$ and $m(t,y)$ w.r.t. $y$. If $\beta\in \bar{k}$ is a root of $D$, then $\beta$ is a special point of $m$ if and only if $\beta$ is a root of $R$. Consequently, $D$ and $\gcd(R,D)$ have the same irreducible special factors.
	\end{lemma}
	\begin{proof}
		Let $\beta\in \bar{k}$ be a root of $D$. Since $D$ is squarefree, we have $\partial_t(D)(\beta)\neq0$.
		As $m(t,y)$ is monic in $y$, the degrees (in $y$) of
		$D_0(\beta,y)$ and $m(\beta,y)$ coincide with those of
		$D_0(t,y)$ and $m(t,y)$, respectively. Then $R(\beta)$ is the resultant of $D_0(\beta,y)$ and $m(\beta,y)$ w.r.t. $y$.
		
		Write $D=(t-\beta)\tilde{D}$. Then $\tilde{D}(\beta)\ne 0$. A calculation similar to~\eqref{EQ:diffroot} yields
		$\kappa(D)(\beta)=-\beta'\tilde{D}(\beta)$ and $\partial_t(D)(\beta)=\tilde{D}(\beta)$.
		Then $D_0(\beta,y)=\tilde{D}(\beta)(y-\beta')$, and $\res_y\bigl(D_0(\beta,y),m(\beta,y)\bigr)=0$ if and only if $m(\beta,\beta')=0$, i.e., $\beta$ is a special point of~$m$.
	\end{proof}
	
	\begin{algorithm}\label{ALG:splitting}
		{\sc SplittingFactorization}
		
		\smallskip \noindent
		{\sc Input:} $D\in k[t]$ and $m$, the monic irreducible polynomial of $t'$.
		
		\smallskip \noindent
		{\sc Output:} the splitting factorization of $D$ w.r.t. $m$.
		
		\begin{itemize}
			\item[1.] Compute the squarefree factorization $D=D_1^{\mu_1}\dots D_n^{\mu_n}$ of $D$
			\item[2.] $D_N\leftarrow1$, $D_S\leftarrow1$
			\item[2.] {\sc for $i$ from $1$ to $n$ do}
			\begin{itemize}
				\item[] $D_0\leftarrow\kappa(D_i)+\partial_t(D_i)y$, $R\leftarrow{\rm resultant}_y\bigl(D_0(t,y),m(t,y)\bigr)$
				\item[] $G\leftarrow\gcd(R,D_i)$
				\item[] $D_S\leftarrow D_S\,G^{\mu_i}$, $D_N\leftarrow D_N \left(\frac{D_i}{G}\right)^{\mu_i}$
			\end{itemize}
			{\sc end do}
			\item[4.] {\sc return} $D_N,D_S$
		\end{itemize}
		
	\end{algorithm}
	
	The correctness is guaranteed by Lemma~\ref{LEM:splitting}. As a by product, we obtain the squarefree factorization of both  $D_N$ and $D_S$.

\end{document}